\documentclass[10pt]{article}                             
\usepackage{amsmath,amsthm,amsfonts,amssymb, latexsym} 

\usepackage[usenames]{color} 
\newcommand \RR {\mathbb R}

\newcommand \Rb {\overline R}
\newcommand \eps \epsilon 
\newcommand \nablab {\overline \nabla}
\newcommand \bei {\begin{itemize}}
\newcommand \eei {\end{itemize}}
\newcommand \mADM {\mathbf m_\text{ADM}} 
\newcommand \be {\begin{equation}}
\newcommand \ee {\end{equation}}
\newcommand \bel {\be\label}
\newcommand \loc {\text{loc}}

\makeatletter
\DeclareFontFamily{OMX}{MnSymbolE}{}
\DeclareSymbolFont{MnLargeSymbols}{OMX}{MnSymbolE}{m}{n}
\SetSymbolFont{MnLargeSymbols}{bold}{OMX}{MnSymbolE}{b}{n}
\DeclareFontShape{OMX}{MnSymbolE}{m}{n}{
    <-6>  MnSymbolE5
   <6-7>  MnSymbolE6
   <7-8>  MnSymbolE7
   <8-9>  MnSymbolE8
   <9-10> MnSymbolE9
  <10-12> MnSymbolE10
  <12->   MnSymbolE12
}{}
\DeclareFontShape{OMX}{MnSymbolE}{b}{n}{
    <-6>  MnSymbolE-Bold5
   <6-7>  MnSymbolE-Bold6
   <7-8>  MnSymbolE-Bold7
   <8-9>  MnSymbolE-Bold8
   <9-10> MnSymbolE-Bold9
  <10-12> MnSymbolE-Bold10
  <12->   MnSymbolE-Bold12
}{}

\let\llangle\@undefined
\let\rrangle\@undefined
\DeclareMathDelimiter{\llangle}{\mathopen}%
                     {MnLargeSymbols}{'164}{MnLargeSymbols}{'164}
\DeclareMathDelimiter{\rrangle}{\mathclose}%
                     {MnLargeSymbols}{'171}{MnLargeSymbols}{'171}
\makeatother

\numberwithin{equation}{section}
        \newtheorem{theorem}{Theorem}[section]
        \newtheorem{proposition}[theorem]{Proposition}
        \newtheorem{lemma}[theorem]{Lemma}
        \newtheorem{corollary}[theorem]{Corollary} 
        \newtheorem{definition}[theorem]{Definition}

\DeclareMathOperator{\End}{End}  
                                                                                
\DeclareMathOperator{\Div}{div}

\DeclareMathOperator{\spin}{Spin}
\DeclareMathOperator{\cl}{Cl}
\newcommand\rr{\mathbb{R}}

\newcommand\D{\mathcal{D}}

\newcommand\too{\longrightarrow}

\newcommand\isom{\cong}

\newcommand\Wloc{W_{\mathrm{loc}}}
\newcommand\Lloc{L_{\mathrm{loc}}}

\begin{document}

\title{The positive mass theorem
for manifolds
\\
 with distributional curvature}

\author{Dan A. Lee\footnote{Graduate Center and Queens College, City University of New York, 365 Fifth Avenue,  New York, NY 10016, USA. Email: {\sl dan.lee@qc.cuny.edu.}} \, 
and Philippe G. LeFloch\footnote{Laboratoire Jacques-Louis Lions, Centre National de la Recherche Scientifique, Universit\'e Pierre et Marie Curie, 4 Place Jussieu, 75252 Paris, France. 
\newline
Email : {\sl contact@philippelefloch.org.}}
} 
\date{August 2014}

\maketitle

\begin{abstract} We formulate and prove a positive mass theorem for $n$-dimensional spin manifolds whose metrics have only the Sobolev regularity $C^0 \cap \Wloc^{1,n}$.  At this level of regularity, the curvature of the metric is defined in the distributional sense only, and we propose here a (generalized) notion of ADM mass for such a metric. Our main theorem establishes that if the manifold is asymptotically flat and has non-negative \emph{scalar curvature distribution}, then its (generalized) ADM mass is well-defined and non-negative, and vanishes only if the manifold is isometric to Euclidian space. Prior applications of Witten's spinor method by Lee and Parker and by Bartnik required the much stronger regularity $\Wloc^{2,2}$.  Our proof is a generalization of Witten's arguments, in which we must treat the Dirac operator and its associated Lichnerowicz-Weitzenb\"ock identity in the distributional sense 
and cope with certain averages of first-order derivatives of the metric over annuli that approach infinity. Finally, we observe that our arguments are not specific to scalar curvature and also allow us to establish a ``universal'' positive mass theorem.
\end{abstract}
 
%\newpage 
%
%\setcounter{tocdepth}{3}
%\tableofcontents
 
%===========================================================================

\section{Introduction}
\label{sec:1}

A fundamental problem in Riemannian geometry is to understand generalized notions of curvature restrictions.  For example, Toponogov's theorem motivates a generalized notion of non-negative sectional curvature that makes sense for length spaces, and the theory of these length spaces of non-negative sectional curvature could be thought of as the gold standard for a theory of ``singular curvature''.  

In this paper we consider the question of whether metrics with non-negative scalar curvature \emph{in the sense of distributions} share any interesting properties with honest-to-goodness $C^2$-regular metrics with non-negative scalar curvature in the classical sense. Our main result (in Theorem \ref{pmt} below) is that the positive mass theorem generalizes to this setting. One reason to consider weak regularity for the positive mass theorem\footnote{After completion of this work, Cox pointed out to us that Theorem \ref{pmt} implies that if a sequence of smooth complete asymptotically flat metrics of non-negative scalar curvature happens to converge in $C^1$ and has mass converging to zero, then that limit space must be Euclidean. By applying this argument he can deduce a topological positive mass stability theorem.} 
 is for application to stability of the positive mass theorem (§cf.~\cite{LS1,LS2,LeFlochSormani} and the references therein). 

Recall that the positive mass theorem was established by Schoen and Yau in dimensions $n$ less than eight \cite{SY, Schoen} and Witten for spin manifolds \cite{Witten}, under the assumption that the underlying metric is regular. (See also \cite{Lohkamp} for some advances in the general case.) 
Bartnik \cite{Bartnik} showed that Witten's spinor argument works whenever the metric is\footnote{We use the standard notation for the Lebesgue spaces $L^p$ and $L_\loc^p$ and the Sobolev spaces $W^{k,p}$ and $W^{k,p}_\loc$.}
 $\Wloc^{2,p}$ with $p>n$. For the slightly weaker integrability class $C^0\cap \Wloc^{2,n/2}$, see \cite{Grant-Tassotti}. As far as solely ``piecewise regular'' metrics are concerned, Miao  \cite{Miao} used a smoothing plus conformal deformation (following Bray \cite{Bray}) and proved a version of the positive mass theorem for metrics that are singular only along a hypersurface. Similar results were also proved by Shi and Tam \cite{ST} (using Witten's spinor method) and McFeron and Székelyhidi \cite{MFS} (using the Ricci flow). The conformal deformation method was also used by Lee  \cite{Lee} to treat metrics with low-dimensional singular sets. 

Our result only assumes that the metric is $C^0\cap \Wloc^{1,n}$ and thereby generalizes \emph{all} of those previous results in the spin case, as explained in Section~\ref{sec:6}. Our result also fits together with and was motivated by earlier work by LeFloch and collaborators \cite{LeFloch,LM,LR,LSt}, who defined and investigated the Einstein equations within the broad class of metrics with $L^\infty \cap \Wloc^{1,2}$ regularity and established existence results for the Cauchy problem at this level of regularity. 

We state here our main result and refer to Section~\ref{sec:2} below for details. 

\begin{theorem}[The positive mass theorem for distributional curvature]
\label{pmt}
Let $M$ be a smooth  $n$-manifold ($n\ge 3$) endowed with a spin structure and a $C^0\cap W^{1,n}_{-q}$ regular and asymptotically flat, Riemannian metric $g$, with  $q\ge(n-2)/2$.  If the \emph{distributional scalar curvature} $R_g$ of $g$ is non-negative, then its \emph{generalized ADM mass}, denoted by $\mADM(M,g)$, is non-negative,  that is, 
$$
\mADM(M,g) \geq 0, 
$$
Moreover, equality occurs only when $(M,g)$ is isometric to Euclidean space.
\end{theorem}

Note that under the conditions in Theorem~\ref{pmt}, the mass $\mADM(M,g)$ exists but could be (positive) infinite; however, we will present a ``finiteness'' condition at infinity that guarantees that the ADM mass is finite. We also point out that once the appropriate spaces are defined, it follows from the Sobolev embedding theorem that $W^{1,p}_{-q} \subset C^0\cap W^{1,n}_{-q}$ for any $p>n$, so that Theorem~\ref{pmt} holds in the class $W^{1,p}_{-q}$  for any $p>n$. 

Our proof of Theorem~\ref{pmt} relies on the following main ideas:
\bei

\item The first difficulty, dealt with in Section~\ref{sec:2}, is defining the notions required in the statement of the theorem, including the concepts of $W^{1,n}_{-q}$ asymptotic flatness, distributional curvature, and generalized ADM mass. 

\item Our notion of distributional scalar curvature (cf.~Section~\ref{sec:21}) is based on a choice of a fixed backgound metric and on a reformulation of the expression of the scalar curvature (following \cite{LM} as well as \cite{GT}). 

\item Our asymptotic flatness condition (cf.~Section~\ref{sec:22}) implies that the manifold is complete in the sense of metric spaces, although Hopf-Rinow theorem does not apply. Since under our low regularity assumptions the connection coefficients are only $\Lloc^n$, the existence theorem for geodesics does not apply; however, since $g$ is continuous, $(M,g)$ can still naturally be seen as a metric space endowed with the distance function $d_g$ induced by this metric. 

\item For the definition of the ADM mass in Section~\ref{sec:22} we introduce averages of first-order derivatives of the metric over {\sl annuli} that approach infinity and we consider their limit at infinity. 

\item Next, in order to extend Witten's argument, we derive a distributional version of the Lichnerowicz-Weitzenb\"ock identity that is valid for compactly supported spinors in $W^{1,2}$ and assumes only $C^0\cap \Wloc^{1,n}$ regularity of the metric.

\item An $L^2$-based setting for spinor field solutions to the Dirac equation is developed here under the low regularity conditions that the metric is solely $C^0\cap W_{-q}^{1,n}$ regular and asymptotically flat. 
 
\eei 

An outline of this paper is as follows. Section~\ref{sec:2} presents our generalized notions of curvature and mass, together with some additional properties of the generalized ADM mass that are not stated in Theorem~\ref{pmt}.  Section~\ref{sec:3} contains the main argument, based on Witten's approach, assuming the existence of suitable spinors, while the latter issue is the subject of  Section~\ref{sec:4}. In the Section \ref{sec:6} we verify that the earlier results in \cite{Miao, ST, MFS, Lee} for ``piecewise regular'' metrics can be recovered from Theorem~\ref{pmt} in the spin case. In a final section we explain how our method can be generalized to the setting of Herzlich's universal positive mass theorem  \cite{Herzlich}.

%===========================================================================

\section{Scalar curvature and ADM mass for $\Wloc^{1,p}$ metrics}
\label{sec:2} 

\subsection{Distributional scalar curvature}
\label{sec:21}

Throughout this paper, we are given a smooth\footnote{It is straightforward to check that $C^\infty$-smoothness is not needed for our definitions and results; in particular, $C^3$-smoothness of $M$ and $C^2$-smoothness of $h$ are sufficient.} $n$-manifold $M$ with $n \geq 3$, on which we define a fixed\footnote{An additional restriction will be imposed when we will define \emph{asymptotic flatness} in Section~\ref{sec:22}.} 
  smooth background metric denoted by $h$. Using this metric, it is a standard matter to define the family of Lebesgue spaces $\Lloc^p$ for $p \in [1, +\infty]$ and Sobolev spaces $\Wloc^{k,p}$, which do not depend upon the choice of $h$ (so long as only local integrability is concerned). 

On $M$, we can consider a general Riemannian metric $g$ in $\Lloc^\infty$ with inverse $g^{-1} \in \Lloc^\infty$. By this, we mean that
$g$ is an inner product at (almost) every point of $M$ and, in any smooth local coordinate chart, $g_{ij}$ and its inverse $g^{ij}$ (coefficient functions that are defined almost everywhere) are also locally bounded. For a metric with such  regularity, one cannot in general define a notion of scalar curvature in the classical way, but following LeFloch and Mardare \cite{LM} (see also \cite{GT}) and provided we further assume that $g\in \Wloc^{1,2}$, one can define the scalar curvature of $g$ as a {\sl distribution,} in the following way.  

In order to justify some preliminary calculations, we assume first that the metric $g$ is sufficiently regular and we define a \emph{tensor} on $M$ by
\bel{eq:gamma}
\Gamma_{ij}^k :=\frac{1}{2}g^{kl}(\nablab_i g_{jl}+\nablab_j g_{il}-\nablab_l g_{ij}),
\ee
where $\nablab$ denotes the Levi-Civita connection of the background metric~$h$. Our notational convention throughout this paper is that barred quantities are defined using the background metric $h$.  

Routine computations yield the following relationship between the scalar curvature $R_g$ of $g$ and the scalar curvature $\Rb $ of~$h$: 
\begin{align*}
R_g&=\Rb +g^{ij}\left(\nablab_k \Gamma^k_{ij}-\nablab_j \Gamma^k_{ki}\right)+g^{ij}\left(\Gamma^k_{k\ell}\Gamma^\ell_{ij}-\Gamma^k_{j\ell}\Gamma^\ell_{ik}\right) \\
&=  \nablab_k \left(g^{ij} \Gamma^k_{ij}- g^{ik} \Gamma^j_{ji}\right)+
\Rb  
\\
& \quad -  \nablab_k g^{ij} \Gamma^k_{ij} + \nablab_k g^{ik} \Gamma^j_{ji}
+ g^{ij}\left(\Gamma^k_{k\ell}\Gamma^\ell_{ij}-\Gamma^k_{j\ell}\Gamma^\ell_{ik}\right).
\end{align*}
Thus, we find the {\bf scalar curvature decomposition} 
\be 
\label{scalar}
R_g = \nablab_k V^k + F, 
\ee
which is determined by the vector field $V^k$ and scalar field $F$    
$$
V^k := g^{ij} \Gamma^k_{ij}- g^{ik} \Gamma^j_{ji},
$$
\bel{eq:207}
\aligned
F & := \Rb  -  \nablab_k g^{ij} \Gamma^k_{ij} + \nablab_k g^{ik} \Gamma^j_{ji}
+ g^{ij}\big(\Gamma^k_{k\ell}\Gamma^\ell_{ij}-\Gamma^k_{j\ell}\Gamma^\ell_{ik}\big). 
\endaligned
\ee
A simple computation allows us to rewrite the vector field as
\bel{defineV}
V^k = g^{ij}g^{k\ell}( \nablab_j g_{i\ell} - \nablab_\ell g_{ij}).
\ee
This calculation motivates us to relax the regularity on the metric by observing that, under low regularity conditions on $g$, the fields 
$\Gamma$, $V$, and $F$ are still well-defined, which allows us to generalize the notion of scalar curvature. 

\begin{definition}
\label{def:22}
Let $M$ be a smooth manifold endowed with a smooth background metric~$h$. 
Given any Riemannian metric $g$ with  $\Lloc^\infty\cap\Wloc^{1,2}$ regularity and locally bounded inverse $g^{-1}\in\Lloc^\infty$, the {\bf scalar curvature distribution} $R_g$ is defined, for every compactly supported  smooth (test-) function $u: M \to \RR$
 by  
\bel{distscalar}
\llangle R_g,  u \rrangle
:= \int_M \left( -V \cdot\nablab \Big( u \, \frac{d\mu_g}{d\mu_h}\Big)  + F \, u \, \frac{d\mu_g}{d\mu_h} \right) d\mu_h, 
\ee
in which the dot product is taken using the metric $h$ and $d\mu_h$ and $d\mu_g$ denote the volume measures associated with $h$ and $g$, respectively, and furthermore 
\item 
\bei

\item $\Gamma, V, F$ are defined above in equations \eqref{eq:gamma}--\eqref{defineV},

\item one has the regularity $\Gamma \in L^2 _\loc$,  $V \in \Lloc^2$, and $F \in \Lloc^1$, and 

\item 
$\frac{d\mu_g}{d\mu_h} \in \Lloc^\infty\cap\Wloc^{1,2}$ is the density of $d\mu_g$ with respect to $d\mu_h$. 

\eei 
\end{definition}

Under the assumptions in Definition~\ref{def:22}, the two terms in the right-hand side of \eqref{distscalar} do make sense: the first term is a product of the form ``$L^2_\loc$ times $L_\loc^2$'', while the second term has the form ``$L^1_\loc$ times $L_\loc^{\infty}$''. 
In the case of sufficiently regular metrics $g$, say of class $C^2$, the scalar curvature $R_g$ is well-defined in the classical way and is a continuous function; in this case, we find $\llangle R_g,  u \rrangle= \int_M R_g  u\,d\mu_g$, and of course this observation motivates our definition \eqref{distscalar}.  
Furthermore, for $C^2$ metrics $g$, the quantity $\llangle R_g,  u \rrangle$ does not depend on the choice of the background metric $h$ and, consequently, it follows (from a standard density argument) that the distribution $\llangle R_g,  u \rrangle$ is also independent of the choice of $h$, as long as $g$ is in $C^0\cap\Wloc^{1,2}$. 

Recall that a distribution such as $R_g$ is said to be non-negative when $\llangle R_g, u\rrangle\geq0$ for every non-negative test function $u$. This allows us to make sense of the phrase \emph{non-negative distributional scalar curvature}, which was used in Theorem~\ref{pmt}.  

Although the scalar curvature distribution is well-defined for any  $g \in \Lloc^\infty\cap\Wloc^{1,2}$, we will have to assume that $g\in C^0\cap \Wloc^{1,n}$ in order to prove a positive mass theorem. The following proposition spells out what sort of test functions may be used with the scalar curvature distribution under this regularity assumption.

\begin{proposition} 
\label{prop:regul1}
Let $M$ be a smooth manifold endowed with a smooth background metric~$h$. Given any Riemannian metric $g$ with 
$C^0\cap \Wloc^{1,n}$ regularity, the scalar curvature distribution $R_g$ (in the sense of Definition~\ref{def:22}) can be extended so that \eqref{distscalar} makes sense and defines $\llangle R_g,  u\rrangle$ for all compactly supported functions $u \in L^{\frac{n}{n-2}}$ whose derivatives lie in  $L^{\frac{n}{n-1}}$.
\end{proposition} 

\begin{proof} 

Our assumptions imply that $V\in L_\loc^n$, $F \in L_\loc^{n/2}$, $\frac{d\mu_g}{d\mu_h}\in C^0$,  and $\nablab\frac{d\mu_g}{d\mu_h} \in L_\loc^n$. We see that the expression in \eqref{distscalar} defining $\llangle R_g,  u\rrangle$ involves integrating 
\[ 
- \Big( V\cdot\nablab u \Big) \frac{d\mu_g}{d\mu_h}- V\cdot u \, \nablab\frac{d\mu_g}{d\mu_h}  + F u \frac{d\mu_g}{d\mu_h}.\]
Each of these terms is integrable because $\frac{1}{n}+ \frac{n-1}{n} =1$, $\frac{1}{n}+\frac{n-2}{n}+\frac{1}{n}=1$, and $\frac{2}{n}+\frac{n-2}{n}=1$, respectively.
\end{proof} 

%------------------------------------------------------------------------------------------------

\subsection{The generalized notion of ADM mass}
\label{sec:22}

We now turn our attention to the definitions of asymptotic flatness and ADM mass. 
Assume that $M$ is a smooth  $n$-manifold such that 
 there exists a compact set $K\subset M$ and a diffeomorphism $\Phi$ between $M\smallsetminus K$ and  $\rr^n\smallsetminus B_1(0)$, where $B_1(0)$ denotes the unit ball in $\rr^n$. This pair $(M, \Phi)$ might be called  a ``topologically asymptotically flat manifold,''  but, given the context of this paper, we will simply call it a {\bf background manifold} for short.
 
Given a background manifold $(M, \Phi)$, choose any smooth background metric $h$ on $M$ such that $h_{ij}=\delta_{ij}$ in the coordinate chart $M\smallsetminus K\isom \rr^n\smallsetminus B_1(0)$ determined by $\Phi$. We also choose a smooth positive function $r$ on $M$ that coincides with the radial coordinate on $M\smallsetminus K\isom \rr^n\smallsetminus B_1(0)$ and is less than $2$ on $K$. Observe that the manifold $(M,h)$ is automatically both geodesically complete (since $K$ is compact) and (consequently) complete as a metric space. We will call the pair $(h,r)$ {\bf background metric data} on $(M, \Phi)$. This data plays no essential role and is used only for the purpose of stating simpler definitions.

Given any $p>0$, $s\in\rr$, we define the weighted space $L^p_s(M)$ of all functions $u$ with finite norm 
$$  
\| u \|_{L^p_s(M)} = \left(\int_M |u|^p\, r^{-ps-n}\,d\mu_h\right)^{1/p}. 
$$
This definition easily extends to tensors and spinors defined on $M$. Next, for positive integers $k$, we introduce the weighted Sobolev space $W^{k, p}_s(M)$ of all functions $u$ with finite norm 
$$ 
\| u \|_{W^{k,p}_s(M)} = \sum_{i=0}^k \| \nablab^i u \|_{W^{p}_{s-i}(M)}, 
$$
with a similar definition for tensors and spinors. Observe that, although the norms depend on the choices of $h$ and $r$, the spaces $L^p_s(M)$ and $W^{k,p}_s(M)$ themselves only depend on the background manifold $(M,\Phi)$, since $h_{ij}=\delta_{ij}$ in the asymptotically flat coordinate chart. 

With this notation, we can now  introduce the following notions. 

\begin{definition} Let $(M,\Phi)$ be a background manifold endowed with background metric data $(h,r)$.
For any $p\ge1$ and $q> 0$, a $\Lloc^\infty$ Riemannian metric $g$ on $M$ with $\Lloc^\infty$ inverse
 is said to be {\bf $W^{k,p}_{-q}$ asymptotically flat} if and only if $g-h \in W^{k,p}_{-q}(T^*M\otimes T^*M)$. 
Furthermore, the {\bf generalized ADM mass} of such a manifold is then defined as 
$$
\mADM(M,g) := \frac{1}{2(n-1)\omega_{n-1}}\inf_{\eps>0} \liminf_{\rho\to +\infty} 
\Bigg( \frac{1}{\eps} \int_{\rho<r<\rho+\eps} V \cdot\nablab r\,d\mu_h \Bigg),
$$
where $V$ is the vector field \eqref{defineV} and $\omega_{n-1}$ is the volume of the standard unit $(n-1)$-sphere.
\end{definition}

Note that these definitions are independent of the specific choice of the data $(h,r)$. Our definition of ADM mass generalizes the usual definition, as will become clear from Corollary \ref{ClassicalMass}, below. At this juncture, we can already see the verisimilitude by looking at equation \eqref{defineV} and observing that we have replaced the usual flux integral by an integral over an annulus. Integrating over an annulus is necessary in our framework, since the assumed regularity is too low to give a meaning to the flux integrals themselves.  For the sake of simplicity in the presentation, the manifold is assumed to have only one asymptotically flat end, but it is straightforward  to extend our definitions and arguments to manifolds with an arbitrary number of asymptotically flat ends. 

%---------------------------------------------------------------------------------------

\subsection{Basic properties of the generalized ADM mass}

Given the definitions above, the statement in Theorem~\ref{pmt} now makes sense. Although our ADM mass is well-defined under the conditions therein, we require a separate assumption to guarantee that the mass is \emph{finite}. This is expected since the classical theory requires the scalar curvature to be integrable in order to have a well-defined mass. The closest analog of integrability for us is to assume that the scalar curvature is a finite signed measure (outside a compact set). This assumption is partly motivated by the fact that we will eventually assume that  $R_g$ is non-negative, which will imply that it is at least a {\sl locally finite} measure \cite{Schwartz}. In other words, in the case where $R_g$ is non-negative, our assumption is only about the finiteness property outside a compact set.

\begin{proposition}
\label{mass}
Let $(M,\Phi)$ be a background manifold endowed with background metric data $(h,r)$.
Let $g$ be a $C^0\cap W^{1,n}_{-q}$ asymptotically flat metric on $M$, with $q> (n-2)/2$.
\begin{enumerate}
\item
If $R_g$ is a \emph{finite, signed} measure outside some compact set, then, for any $\eps>0$, the limit
\bel{MassAlt}
 m:=\frac{1}{2(n-1)\omega_{n-1}} \lim_{\rho\to +\infty} \Bigg( \frac{1}{\eps} \int_{\rho<r<\rho+\eps} V \cdot\nablab r\,d\mu_h \Bigg) 
 \ee
exists, is finite, and does not depend on $\eps$.
\item
If $R_g$ is a measure outside some compact set, then equation \eqref{MassAlt} holds and is independent of $\eps$, though possibly infinite. Moreover, in this case, the mass is finite if and only if $R_g$ is a  \emph{finite} measure outside a compact set.
\end{enumerate}
\end{proposition}
Note that this proposition concerning finiteness of the mass requires $q>(n-2)/2$ while Theorem~\ref{pmt} only requires $q\ge (n-2)/2$.

\begin{proof} Given any $\eps>0$ and and $\rho>2$, we consider the following cut-off function associated with the function $r$: 
\bel{cut-off}
 \chi_\rho(x) = 
\left\{ \begin{array}{ll}
1, \qquad & r(x)\le\rho,
\\
1 + \frac{1}{\eps}(\rho-r(x)), & \rho<r(x)\le\rho+\eps,
\\
0, & \rho+\eps \leq r(x).
\end{array}\right.
\ee
Since $\chi_\rho$ is a compactly supported Lipschitz continuous function, Proposition~\ref{prop:regul1}  implies that we may use $\chi_\rho \frac{d\mu_h}{d\mu_g}$ as a test function for the scalar curvature distribution to obtain:
\begin{align}
\left\llangle R_g, \chi_\rho \frac{d\mu_h}{d\mu_g} \right\rrangle
& = \int_M \left( -V \cdot\overline\nabla \chi_\rho + F \chi_\rho\right)\,d\mu_h \nonumber\\
& = \int_{\big\{\rho<r<\rho+\eps\big\}} V \cdot \tfrac{1}{\eps}\overline\nabla r \,d\mu_h
+\int_M F \chi_\rho\,d\mu_h. \label{MassLimit}
\end{align}
Now, a simple computation shows that if $q>(n-2)/2$, then $L^{n/2}_{-2q-2}\subset L^1$, and thus $F$ is integrable. Hence,  since $\chi_\rho$ pointwise approaches $1$, Lebesgue's dominated convergence theorem implies that $\int_M F \chi_\rho\,d\mu_h$ converges (as $\rho\to +\infty$) to $\int_M F \, d\mu_h$.

To prove Part 1 of the proposition, assume that $R_g$ is a finite signed measure outside a compact set. Then dominated convergence (for integration with respect to the signed measure $R_g$) and boundedness of $ \frac{d\mu_h}{d\mu_g}$ implies that $\left\llangle R_g, \chi_\rho \frac{d\mu_h}{d\mu_g} \right\rrangle$ converges as $\rho\to +\infty$. So by taking limits in \eqref{MassLimit}, we see that
\bel{MassAlt2}
\lim_{\rho\to +\infty} \frac{1}{\eps} \int_{\big\{\rho<r<\rho+\eps\big\}} V \cdot \overline\nabla r \,d\mu_h = \left\llangle R_g,  \frac{d\mu_h}{d\mu_g} \right\rrangle - \int_M F \,d\mu_h 
\ee
exists, is finite, and does not depend on $\eps$.

To prove Part 2 of the proposition, we now assume that $R_g$ is a measure outside a compact set.
Then $\left\llangle R_g, \chi_\rho \frac{d\mu_h}{d\mu_g} \right\rrangle$ is {\sl monotone} in $\rho$, and so it must have a limit (possibly infinite) as $\rho\to +\infty$. Once again equation \eqref{MassAlt2} holds, and the left-hand side (the mass) is finite if and only if $\left\llangle R_g, \frac{d\mu_h}{d\mu_g} \right\rrangle$ is finite if and only if  $R_g$ is a finite measure outside a compact set, since $\frac{d\mu_h}{d\mu_g}$ is bounded above and below by a positive number.
\end{proof}

From our definition and assuming sufficient regularity on the metric, we recover the definition in Bartnik \cite{Bartnik}.

\begin{corollary}\label{ClassicalMass}
Let $g$ be a $W^{2,p}_{-q}$ asymptotically flat metric with $p>n$ and $q\ge (n-2)/2$, and assume that the scalar curvature of $g$ is integrable. Then the generalized ADM mass coincides with the standard ADM mass.
\end{corollary}

Observe here that if $g\in \Wloc^{2,p}$ with $p\ge1$, then $R_g\in \Lloc^1$, and hence $R_g$ (thought of as a distribution) being a finite measure outside a compact set is then equivalent to $R_g$ (thought of as a function) being integrable.

\begin{proof}
Under our assumption on $g$, the vector field $V \in W^{1,p}_{-q-1}$ and, in particular, is H\"{o}lder continuous. Therefore, for each $\rho>2$, we have 
$$
\lim_{\eps\to 0} \frac{1}{\eps} \int_{\rho<r<\rho+\eps} V \cdot \overline\nabla r \,d\mu_h 
= \int_{r=\rho} V \cdot \overline\nabla r \,d\sigma_h, 
$$
where $d\sigma_h$ is the induced volume measure on the sphere. Arguing as in the proof of Proposition~\ref{mass}, one can see that the convergence of $\frac{1}{\eps} \int_{\rho<r<\rho+\eps} V \cdot \overline\nabla r \,d\mu_h$ as $\rho\to +\infty$ is uniform in $\eps$, and since the limit (as $\eps\to0$) exists, we can change the order of the limits and obtain
\be
\label{eq:209bl}
\aligned
m
& ={1\over 2(n-1)\omega_{n-1}} \lim_{\rho\to +\infty} \int_{r=\rho} V \cdot \overline\nabla r \,d\sigma_h
\\
&={1\over 2(n-1)\omega_{n-1}}\lim_{\rho\to +\infty}\int_{r=\rho} \sum_{i,j=1}^n (g_{ij,i}-g_{ii,j})\nu_j d\sigma_h,
\endaligned
\ee
which is the usual definition of ADM mass, where  $\nu = \nablab r$ is the Euclidean outward unit normal. (Recall that $h=\delta$ in the exterior region.) The second equality in \eqref{eq:209bl} follows easily from equation \eqref{defineV} and the fact that $g_{ij}-\delta_{ij}=o(r^{-q})$ and $g_{ij,k}=o(r^{-q-1})$, via the weighted Sobolev embedding theorem \cite{PT}. 
\end{proof}

%=====================================================================

\section{Witten's argument with distributional curvature} 
\label{sec:3}

\subsection{Lichnerowicz--Weitzenb\"{o}ck formula for the Dirac operator}
\label{sec:31}

We are now in a position to present the core of the proof of Theorem \ref{pmt} in this section, while some technical material is postponed to the following section. We assume that $(M,\Phi)$ is a background $n$-manifold with a spin structure, endowed with background metric data $(h,r)$, and we suppose that $g$ is a $W^{1,n}_{-q}$ regular and asymptotically flat metric on $M$ with $q=(n-2)/2$. In view of the statement in Theorem \ref{pmt}, it is sufficient to assume the \emph{equality} $q=(n-2)/2$, and so we make this assumption from here on.

We are going to follow Witten's spinor proof \cite{Witten}, borrowing heavily from the expositions in Bartnik \cite{Bartnik} and Parker and Taubes \cite{PT}, while introducing additional arguments along the way, as will be required to cope with metrics with low regularity. 

We start by introducing a Dirac spinor bundle $S$ over $M$, constructed as follows. Start with an irreducible representation $\tau$ of the Clifford algebra $\cl(n)$, which is also a representation of $\spin(n)\subset\cl(n)$. The assumption that $M$ is spin means that there is a principal $\spin(n)$ bundle $E$ over $M$ that double covers the continuous frame bundle determined by $g$. Then, by definition, {\bf spinor bundle} $S$ is the bundle associated with $E$ with fiber determined by the representation $\tau$. Note that this construction implicitly defines an action of the Clifford bundle over $M$ on~$S$. 

In our setup, the role of the background metric must be clarified, as follows. 
Note that since the continuous metric $g$ is homotopic to the smooth background metric $h$, the topology of $S$ is independent of choice of $g$ as follows. There exists a unique self-adjoint isomorphism $b:TM\too TM$ with the property that $g(b^2 v, w)=h(v,w)$ for any $v,w\in TM$. This $b$ maps $h$-frames to $g$-frames. Since this map is homotopic to the identity, it lifts to a map from a principal spin bundle for $h$ to a principal spin bundle for $g$, thereby inducing an isomorphism $\beta$ from the spinor bundle for $h$, which we will simply call $S$, to the spinor bundle $S_g$ for $g$. For convenience, we use this isomorphism to \emph{identify} the two spinor bundles. That is, we define a smooth vector bundle $S$ using $h$ and then define the relevant $g$ quantities via pullback by $\beta$. 

In particular, if we pullback the Hermitian metric on $S_g$ to $S$, we just get the same Hermitian metric on $S$ that comes from $h$, so there is only one relevant Hermitian metric on $S$. We define the {\bf Clifford action $\tau$ via $g$} on $S$ by the formula
\[ 
\tau( bv)\psi = \bar{\tau}(v)\psi, 
\]
where $\bar{\tau}$ is the Clifford action via $h$. We can define the spin connection via $g$ on $S$ by $\nabla_{v} \psi := \beta^{-1} ( \nabla_v (\beta\psi))$, where the second $\nabla$ is the usual spin connection of $g$ on $S_g$. (From here on we will make no more mention of $S_g$ and work exclusively on $S$.)

We can compare the two spin connections as follows. If $\bar{e}_1,\ldots,\bar{e}_n$ is a local orthonormal frame for the metric $h$, then it lifts to a local orthonormal basis $\psi_1,\ldots,\psi_N$ of $S$ (with $N:=2n$).
 We say that these spinors are {\bf constant spinors} with respect to the given frame.
For each of these constant spinors, the {\bf spin connection of $h$} is given by
\bel{SpinConnection-h}
 \nablab\psi_\alpha = \tfrac{1}{4}\sum_{i, j =1 \atop i\neq j}^n \bar{\omega}_{ji} \otimes \bar{\tau}(\bar{e}_i \bar{e}_j)  \psi_\alpha,
 \ee
where $\bar{\omega}_{ji}$ are the connection $1$-forms of the metric $h$ with respect to the chosen basis $\bar{e}_1,\ldots,\bar{e}_n$. Note that the constant spinors $\psi_1,\ldots,\psi_N$ are also constant with respect to $\nabla$, the spin connection of $g$, so that  we have 
\bel{SpinConnection}
 {\nabla}\psi_\alpha = \tfrac{1}{4}\sum_{i, j =1 \atop i\neq j}^n {\omega}_{ji} \otimes {\tau}(e_i e_j)  \psi_\alpha,
 \ee
where $e_i:=b\bar{e}_i$, and
${\omega}_{ji}$ are the connection $1$-forms of the metric $g$ with respect to the chosen basis ${e}_1,\ldots,{e}_n$.

Furthermore, using our local frame again, the {\bf Dirac operator} (with respect to $g$) is defined on a spinor $\psi$ by the standard formula
\bel{eq:Dirac}
\D \psi = \sum_{i=1}^n \tau(e_i) \nabla_i \psi.
\ee
When there is no chance of confusion, we will suppress the variable $\tau$ in our notation.
Observe that since the metric $g$ is only assumed to be $C^0\cap W^{1,n}_\loc$, even if $\psi$ is smooth, one can only conclude that 
the left-hand sides of  \eqref{SpinConnection} and \eqref{eq:Dirac} have {\sl $L^n_\loc$ regularity.}  

One can define weighted Sobolev norms on $S$ using either $h$ or $g$. Although the Hermitian metric is the same, the difference comes from the fact that the volume measures are different and the spin connections are different. 

\begin{lemma}\label{equivalent}
Let $(M,\Phi)$ be a background $n$-manifold endowed with background metric data $(h,r)$.
Let $g$ be a $C^0\cap W^{1,n}_{-q}$ asymptotically flat metric on $M$, with $q=(n-2)/2$. Then the two $W^{1,2}_{-q}$ norms on $S$, defined with respect to $g$ and $h$, are equivalent.
\end{lemma} 
\begin{proof}
Since $g$ and $h$ are uniformly bounded by each other, it is clear that the weighted $L^p$ norms defined with respect to $g$ and $h$ are equivalent. The issue is that the weighted $W^{1,2}$ norms involve derivatives of the metric. (Observe that, for $k>1$, the weighted $W^{k,p}$ norms with respect to $g$ are not even well-defined.)

The difference between the two spin connections $\nabla$ and $\nablab$ on $S$ must be some $\textrm{End}(S)$ valued one-form, say $A$. By our regularity assumption on $g$, $A$ must be $L^n_{-q-1}$ integrable, which can be seen explicitly from formulas \eqref{SpinConnection-h} and \eqref{SpinConnection}. For any smooth spinor $\psi$ and using an obvious notation for the norms, we compute
\begin{align*}
\| \psi \|_{W^{1,2}_{-q}(g)} &= \|  \psi \|_{L^2_{-q}(g)} + \| \nabla\psi \|_{L^2_{-q-1}(g)} \\ 
&= \|  \psi \|_{L^2_{-q}(g)} + \| \nablab\psi +A\psi \|_{L^2_{-q-1}(g)} \\ 
&\le \|  \psi \|_{L^2_{-q}(g)} + \| \nablab\psi \| _{L^2_{-q-1}(g)} 
 +\| A\psi \|_{L^2_{-q-1}(g)} \\ 
&\le C \| \psi \|_{W^{1,2}_{-q}(h)}  + C\| A\psi \|_{L^2_{-q-1}(h)} 
\end{align*}
for some constant $C>0$.
We estimate the latter term using the weighted H\"{o}lder inequality, the weighted Sobolev inequality, and regularity of $A$: 
\begin{align*}
 \| A\psi \|_{L^2_{-q-1}(h)} & \le \| A \|_{L^n_{-1}(h)}\,  \| \psi \|_{L^{\frac{2n}{n-2}}_{-q}(h)} \\ 
 &\le C   \left\| |\psi| \right\|_{W^{1,2}_{-q}(h)}
\\ 
&\le C   \|  \psi  \|_{W^{1,2}_{-q}(h)},    
\qedhere
\end{align*}  
where the last inequality follows from Kato's inequality.
\end{proof}

Recall that, in the standard smooth case, the Lichnerowicz--Weitzenb\"{o}ck formula reads 
\bel{eq:Lichne}
\D^2 \psi = \nabla^*\nabla \psi + \tfrac{1}{4}R_g \psi \quad \text{ (for sufficiently smooth $g, \psi$),}
\ee
where $\nabla^*$ is the formal adjoint of $\nabla$.
The main idea in Witten's proof of the positive mass theorem is to use an asymptotically constant
 solution $\psi$ of the Dirac equation, and then integrate the Lichnerowicz--Weitzenb\"{o}ck formula (for $\psi$) against $\psi$ itself. Using the non-negativity of $R_g$, Witten obtains a boundary integral at infinity with a sign, and then shows that the boundary integral is just the ADM mass. For our theorem, $g$ and $\psi$ \emph{cannot be differentiated twice,} and so we must integrate by parts \emph{before} applying the Lichnerowicz--Weitzenb\"{o}ck formula. Moreover, since the boundary integrals do not make sense, we need a cut-off function to mimic the standard behavior.

First of all, for any {\sl smooth} metric $g$ and {\sl smooth}  spinor $\psi$, integration of the Lichnerowicz--Weitzenb\"{o}ck formula against an arbitrary compactly supported smooth spinor $\phi$ followed by integration by parts yields
$$
 \int_M \langle \D\psi ,\D\phi\rangle\,d\mu_g =\int_M\left( \langle\nabla\psi, \nabla\phi\rangle + 
 {1 \over 4} R_g\langle \psi,\phi\rangle\right)\,d\mu_g,
$$
which we rewrite as
\bel{preLW}
0= \int_M \hskip-.12cm \left( - \langle \D\psi ,\D\phi\rangle +  \langle\nabla\psi, \nabla\phi\rangle
  -  {1 \over 4} \frac{d\mu_h}{d\mu_g} V \cdot\nablab \left(\langle\psi,\phi\rangle \frac{d\mu_g}{d\mu_h}\right)+ 
 {1 \over 4} F\langle \psi,\phi\rangle \right) d\mu_g.
\ee
This suggests the following formula that we will need for our proof.

\begin{proposition}[A Lichnerowicz-Weitzenb\"ock identity for metrics with distributional curvature]\label{distLW}
Assume that $g$ is a $C^0\cap \Wloc^{1,n}$ metric on a smooth $n$-manifold $M$. If $\psi$ and $\phi$ are $\Wloc^{1,2}$ spinors and $\phi$ has compact support, then  
\bel{LW}
0= - \langle \D\psi ,\D\phi\rangle_{L^2} +  \langle\nabla\psi, \nabla\phi\rangle_{L^2}
+ {1 \over 4} \left\llangle R_g, \langle\psi,\phi\rangle\right\rrangle,
 \ee
 where all quantities are computed using $g$.
\end{proposition}
 
\begin{proof} Once we establish that our hypotheses are strong enough to make sense of the equation \eqref{preLW}, the result in the proposition follows from a standard density argument. It is easy to see that $\D\psi$, $\D\phi$, $\nabla\psi$, and $\nabla\phi$ are all in $\Lloc^2$, and so the first two terms of the integrand are well-defined. Note also that $\psi, \phi\in \Wloc^{1,2}\subset \Lloc^{\frac{2n}{n-2}}$ by the Sobolev embeding theorem, and it easily follows from H\"{o}lder's inequality that $\langle\psi,\phi\rangle$ is in $\Lloc^{\frac{n}{n-2}}$ with derivatives in $\Lloc^{\frac{n}{n-1}}$. So by Proposition \ref{prop:regul1}, we see that $\left\llangle R_g, \langle\psi,\phi\rangle\right\rrangle$ is well-defined.

To apply the density argument, we first consider the case when $\psi$ and $\phi$ are smooth, but $g$ is not necessarily smooth. We choose a sequence of $C^2$ metrics $g_i$ converging to $g$ in $C^0 \cap W^{1,n}(K')$, where $K'$ denotes the support of $\phi$. We know that \eqref{LW} holds for each $g_i$, so we just need to show that we can take the limit of this equation to obtain \eqref{LW} for $g$. From formula \eqref{SpinConnection}, it is clear that the difference between the spin connections of $g_i$ and $g$ converges to zero in $L^n(K')$, so that the first two terms of \eqref{LW} converge as desired. The last term converges as desired, because $V_i$ converges to $V$ in $L^n(K')$ and $F_i$ converges to $F$ in $L^{n/2}(K')$.

For the general case, we choose a sequence of $C^2$ spinors $\phi_i$ converging to $\phi$ in $W^{1,2}$, such that the sequence is supported in some fixed compact set $K'$, and also a sequence of $C^2$ spinors $\psi_i$ converging to $\psi$ in $W^{1,2}(K')$. Note that the smoothness property is well-defined since $S$ is a smooth vector bundle, and there exist such sequences using the $W^{1,2}$ norm defined with respect to $h$. 
The calculations in the first paragraph above, and in the proof of Proposition \ref{prop:regul1}, show that if we take the limit of \eqref{preLW} for these $C^2$ spinors, we obtain \eqref{LW}.  To do this, we implicitly use Lemma \ref{equivalent}.
 \end{proof}

%--------------------------------------------------------------------------------------------------

\subsection{Witten identity for distributional curvature}
\label{sec:32}

Choose $\eps>0$ and $\rho>2$, and let $\chi_\rho$ be the cut-off function \eqref{cut-off} defined in the proof of Proposition~\ref{mass}. Using $\chi_\rho\psi$ as our $\phi$ in equation (\ref{LW}), we can mimic the creation of the boundary term that appears in Witten's argument, as follows. By introducing the operator 
$$ 
L_i\psi:= \sum_{1\le j\le n \atop j\neq i} e_i e_j  \nabla_j \psi,
$$
we obtain the following identity.

\begin{lemma}\label{Lem:BoundaryTerm}
Let $(M,\Phi)$ be a background $n$-manifold endowed with background metric data $(h,r)$, and assume that $g$ is a $C^0\cap \Wloc^{1,n}$ metric on $M$. Then, for every spinor field $\psi\in\Wloc^{1,2}(S)$ and all $\eps>0$ and $\rho>2$, one has 
\bel{BoundaryTerm}
\aligned
& \frac{1}{\eps}\int_{\rho<r<\rho+\eps}
 \sum_{i=1}^n  \langle L_i\psi , \psi \rangle \nabla_i r \,d\mu_g
\\
& = -\langle \D\psi ,\chi_\rho \D\psi\rangle_{L^2}  +  \langle\nabla\psi, \chi_\rho\nabla\psi\rangle_{L^2}
+  {1 \over 4} \llangle R_g, \langle\psi,\chi_\rho\psi\rangle\rrangle.
\endaligned 
\ee
\end{lemma}

\begin{proof} Using \eqref{LW} and for any $\psi\in \Wloc^{1,2}(S)$, we find 
\begin{align*}
 0&=- \langle \D\psi ,\D(\chi_\rho\psi)\rangle_{L^2} +  \langle\nabla\psi, \nabla(\chi_\rho\psi)\rangle_{L^2}
+  {1 \over 4} \left\llangle R_g, \langle\psi,\chi_\rho\psi\rangle\right\rrangle \\
&=  -\langle \D\psi ,\chi_\rho \D\psi\rangle_{L^2}
-   \sum_{i=1}^n \left\langle \D\psi, (\nabla_i \chi_\rho) e_i  \psi \right\rangle_{L^2} \\
&\quad 
 +  \langle\nabla\psi, \chi_\rho\nabla\psi\rangle_{L^2}
 +  \langle\nabla \psi, (\nabla \chi_\rho)\psi\rangle_{L^2}
+  {1 \over 4} \llangle R_g, \langle\psi,\chi_\rho\psi\rangle\rrangle,
\end{align*}thus 
$$
\aligned
0 &= -\langle \D\psi ,\chi_\rho \D\psi\rangle_{L^2}  +  \langle\nabla\psi, \chi_\rho\nabla\psi\rangle_{L^2}
+  {1 \over 4} \llangle R_g, \langle\psi,\chi_\rho\psi\rangle\rrangle\\
&\quad
+\frac{1}{\eps}\int_{\rho<r<\rho+\eps}
 \sum_{i=1}^n \left( \langle \D\psi ,(\nabla_i r) e_i  \psi \rangle -\langle \nabla_i\psi, (\nabla_i r)\psi\rangle
 \right)\,d\mu_g, 
\endaligned
$$
and therefore 
\begin{align*} 
&-\langle \D\psi ,\chi_\rho \D\psi\rangle_{L^2}  +  \langle\nabla\psi, \chi_\rho\nabla\psi\rangle_{L^2}
+  {1 \over 4} \llangle R_g, \langle\psi,\chi_\rho\psi\rangle\rrangle\\
&= \frac{1}{\eps}\int_{\rho<r<\rho+\eps}
 \sum_{i,j=1}^n \left( \langle e_j  \nabla_j \psi , e_i  \psi \rangle -\langle \delta_{ij}\nabla_j\psi, \psi\rangle
 \right)\nabla_i r \,d\mu_g\\ 
&=\frac{1}{\eps}\int_{\rho<r<\rho+\eps}
 \sum_{i,j=1}^n  \langle -(e_i e_j+\delta_{ij})  \nabla_j \psi , \psi \rangle \nabla_i r \,d\mu_g.
\end{align*}
\end{proof}

Let $\partial_1,\ldots,\partial_n$ be the standard basis associated with the asymptotically flat coordinate chart. By applying the Gram-Schmidt method to this basis, we find a $g$-orthonormal frame denoted by $e_1,\ldots, e_n$. As described earlier, this frame lifts to an orthonormal basis of constant spinors. If $\psi$ approaches a constant spinor at infinity sufficiently fast, then the left-hand side of equation \eqref{BoundaryTerm} differs from the integral in equation \eqref{mass} only by the integral of a (sufficiently) decaying function.
More precisely, we have the following lemma.

\begin{lemma} \label{BoundaryLimit}
Assume the hypotheses of Lemma \ref{Lem:BoundaryTerm}, and also assume that $g$ is $W^{1,n}_{-q}$ asymptotically flat with $q=(n-2)/2$.
Let $\psi_0$ be a constant spinor defined over the asymptotically flat coordinate chart, as described above. Then for  any $\rho>2$, $\eps>0$, and any spinor $\psi$ with $\psi-\psi_0 \in W^{1,2}_{-q}(S)$, one has 
$$
  \frac{1}{\eps}\int_{\rho<r<\rho+\eps}
 \sum_{i=1}^n  \langle L_i \psi , \psi \rangle \nabla_i r \,d\mu_g
= \frac{|\psi_0|^2}{4\eps}\int_{\rho<r<\rho+\eps} V\cdot\overline\nabla r\,d\mu_h + \int_{\rho<r<\rho+\eps} u\,d\mu_h
$$
for some function $u\in L^1_{-2q-1}$.
\end{lemma}

\begin{proof}
Assume the hypotheses in the lemma and, for convenience, let us use the notation $O(L^1_{-2q-1})$ for an unnamed function in $L^1_{-2q-1}$. Consider the spinor $\xi := \psi-\psi_0\in W^{1,2}_{-q} $ and the annulus $\Omega := \{x\,|\,  \rho<r(x)<\rho+\eps\}$. Let $I$ be the left-hand side of equation \eqref{BoundaryTerm} that we would like to simplify. There are a few facts that we will need for the computation. First, observe that the Clifford action of $e_i e_j$ with $i\neq j$ on spinors is skew-Hermitian. Also note that the $g$-frame $e_1,\ldots,e_n$ and the coordinate basis $\partial_1,\ldots,\partial_n$ differ, but one has the decay property $e_i-\partial_i \in W^{1,n}_{-q}$. 
%%%%%%%%%%%%%%% 
Using this property, one can compute the $1$-form connections $ \omega_{\ell k}$ 
\bel{ConnectionForm}
 \omega_{\ell k}(e_j) = \tfrac{1}{2}(g_{j\ell,k}-g_{jk,\ell}) + O(L^{n/2}_{-2q-1}), 
 \ee
where the right-hand quantities are computed using the coordinate basis.

Next, we compute
\begin{align*}
\eps I&:= \int_{\Omega}
 \sum_{i=1}^n  \langle L_i\psi , \psi \rangle \nabla_i r \,d\mu_g \\
& = \int_{\Omega}
 \sum_{i=1}^n  \left( \langle L_i\psi_0 , \psi_0 \rangle + \langle L_i\xi , \psi_0 \rangle + \langle L_i\psi , \xi \rangle \right)  \nabla_i r \,d\mu_g \\
& = \int_{\Omega}
 \sum_{i=1}^n  \Big(  \left\langle L_i \psi_0 , \psi_0 \right\rangle + \sum_{ 1\le j\le n \atop j\neq i} \left\langle e_ie_j \nabla_j \xi , \psi_0 \right\rangle +   O(L^1_{-2q-1}) \Big)  \nabla_i r \,d\mu_g, 
\end{align*}
thus 
\begin{align*}
\eps I&
 = \int_{\Omega}
 \sum_{i=1}^n  \Big(  \left\langle L_i\psi_0 , \psi_0 \right\rangle - \sum_{ 1\le j\le n \atop j\neq i} \langle\nabla_j \xi ,  e_i e_j \psi_0 \rangle +   O(L^1_{-2q-1}) \Big)  \nabla_i r \,d\mu_g \\
& = \int_{\Omega}
 \sum_{i=1}^n    \left\langle L_i \psi_0 , \psi_0 \right\rangle  \nabla_i r \,d\mu_g
 + \int_\Omega O(L^1_{-2q-1})\,d\mu_h \\
 &\quad
 +\int_{\Omega}
 \sum_{\substack{i,j=1 \\  i\neq j}}^n  \Big( 
   - \nabla_j (  \langle \xi ,  e_i e_j \psi_0 \rangle \nabla_i r) +  \langle \xi ,  e_i e_j \nabla_j \psi_0 \rangle \nabla_i r
\\
& \hskip2.5cm + \langle \xi ,  e_i e_j \psi_0 \rangle\nabla_j \nabla_i r\Big) \,d\mu_g, 
\end{align*}
where we started to integrate the second term by parts. So, we find  
\begin{align*}
\eps I
& = \int_{\Omega}
 \sum_{i=1}^n    \left\langle L_i \psi_0 , \psi_0 \right\rangle  \nabla_i r  \,d\mu_g   
\\
& \quad +\int_\Omega O(L^1_{-2q-1}) \,d\mu_h
   -\int_{\partial\Omega}  \sum_{\substack{i,j=1 \\  i\neq j}}^n  \langle \xi ,  e_i e_j \psi_0 \rangle (\nabla_i r) \nu_j\,d\sigma_g,
\end{align*}
where the last integral was simplified using the  divergence theorem (for Sobolev functions) on the 
first term, the decay of the second term, and the anti-symmetry (in order to see that the last term vanishes). Next, by the equation \eqref{SpinConnection} and by using anti-symmetry on the last term( since $\nu=\nabla r$), it follows that 
$$
\aligned
\eps I 
& = \int_{\Omega}
\sum_{\substack{i,j,k,\ell=1 \\ i\neq j, k\neq\ell}}^n    \left\langle \tfrac{1}{4} \omega_{\ell k}(e_j)e_i e_j e_k e_\ell  \psi_0 , \psi_0 \right\rangle  \nabla_i r  \,d\mu_g 
  +\int_\Omega O(L^1_{-2q-1}) \,d\mu_h
\\
& = \int_{\Omega}\Big(
\sum_{i,j,k,\ell=1 }^n  \tfrac{1}{8}(g_{j\ell,k}-g_{jk,\ell})(-\delta_{ik}\delta_{j\ell}+\delta_{i\ell}\delta_{jk})|\psi_0|^2
\\
&\qquad \qquad
+   \sum_{i,j,k,\ell\text{ distinct}}^n  \left\langle  \tfrac{1}{8}(g_{j\ell,k}-g_{jk,\ell})e_i e_j e_k e_\ell  \psi_0 , \psi_0 \right\rangle  
\Big)\nabla_i r  \,d\mu_g 
\\
& \quad 
  +\int_\Omega O(L^1_{-2q-1}) \,d\mu_h
\endaligned
$$
by equation \eqref{ConnectionForm}, and by using the skew-Hermitian property, 
$$
\aligned
\eps I 
&= \frac{|\psi_0|^2 }{4}\int_{\Omega}
 \sum_{i,j=1}^n (g_{ij,j}-g_{jj,i})\nabla_i r  \,d\mu_g 
  +\int_\Omega O(L^1_{-2q-1}) \,d\mu_h, 
\endaligned
$$
where the second term vanished by anti-symmetry. Finally, we obtain 
$$
\aligned 
\eps I 
&= \frac{|\psi_0|^2 }{4} \int_{\Omega} 
V \cdot \overline\nabla r  \,d\mu_h 
 +\int_\Omega O(L^1_{-2q-1}) \,d\mu_h,
\endaligned
$$
where we used equation \eqref{defineV} and the decay of the error term in the last line.
\end{proof}

%==============================================================================

\section{Dirac equations with distributional curvature}
\label{sec:4}

\subsection{A technical property}
\label{sec:41}

The technical step in Witten's proof of the positive mass theorem is to find an asymptotically constant spinor that solves the Dirac equation and we need to revisit this construction when the metric has very low regularity. First, we prove a simple technical lemma. 

\begin{lemma}\label{subconverge}
Given a background metric $h$ for asymptotically flat manifold, if $u\in L^1_{-2q-1}$ where $q=(n-2)/2$,  
then for any $\eps>0$, there exists a sequence $\rho_i\to +\infty$ such that
$$
\lim_{i\to +\infty}  \int_{\rho_i<r<\rho_i+\eps} |u|\,d\mu_h =0,
$$
that is, the integral of $|u|$ over the annuli sub-converges to zero.
\end{lemma}

\begin{proof}
Our assumption on $u$ is that $\int_M |u| r^{-1}\,d\mu_h <+\infty$. Define $\rho_i = 2+ i\eps$. Then we have
$$ \sum_{i=1}^\infty  \int_{\rho_i<r<\rho_i+\eps} |u| r^{-1}\,d\mu_h <+\infty.$$
Choose any positive sequence $\delta_i$ whose series diverges slower than the harmonic series in the sense that 
$\sum_{i=1}^\infty \delta_i=\infty$ while $(\rho_i +\eps)\delta_i \to 0$. Since any convergent series must be smaller than any divergent series on an infinite number of terms, there exists a subsequence such that 
 $$\int_{\rho_{i_j}<r<\rho_{i_j}+\eps} |u| r^{-1}\,d\mu_h < \delta_{i_j}.
$$
So we can write
\begin{align*}
 \int_{\rho_{i_j}<r<\rho_{i_j}+\eps} |u|\,d\mu_h 
 &=  \int_{\rho_{i_j}<r<\rho_{i_j}+\eps}r |u| r^{-1} \,d\mu_h \\
&\le (\rho_{i_j} +\eps) \int_{\rho_{i_j}<r<\rho_{i_j}+\eps} |u| r^{-1} \,d\mu_h
\le (\rho_{i_j} +\eps) \delta_{i_j},
\end{align*}
which converges to zero.
\end{proof}

%----------------------------------------------------------------------------------------

\subsection{Existence result for the Dirac equation}
\label{sec:42}
Recall from Lemma \ref{equivalent} that $W^{1,2}_{-q}(S)$, as well as  $L^{2}_{-q-1}(S)$, can be defined using either $g$ or $h$. Throughout this section, we find it more convenient to use the volume form associated with the metric $g$. 
The following proposition allows us to solve the Dirac equation.

\begin{proposition}\label{DiracIsomorphism}
Under the assumptions in Theorem \ref{pmt} and with $q=(n-2)/2$, the operator
$$
\D : W^{1,2}_{-q}(S)\too L^{2}_{-q-1}(S)
$$
is an isomorphism. Note that  $L^{2}_{-q-1}(S)=L^2(S)$ with this choice of $q$.
\end{proposition}

Elliptic estimates for such an operator, assuming our level of regularity, are described in great detail and generality in the work of Bartnik and Chru\'{s}ciel \cite{Bartnik-Chrusciel}. However, in this case we can make a more direct argument, with
the Lichnerowicz-Weitzenb\"ock identity, combined with our assumption of non-negative scalar curvature, providing the necessary estimate.

\begin{proof}
First observe that the operator is clearly a well-defined bounded linear operator. Next, we will prove an injectivity estimate. Given any $\phi\in W^{1,2}_{-q}(S)$ and applying Lemma \ref{Lem:BoundaryTerm} with $\eps=1$ (though the choice of $\eps$ does not matter) to $\phi$, we have
\begin{align*}
& \int_{\rho<r<\rho+1}
 \sum_{i=1}^n  \langle L_i\phi , \phi \rangle \nabla_i r \,d\mu_g
\\
 &= -\langle \D\phi ,\chi_\rho \D\phi\rangle_{L^2}  +  \langle\nabla\phi, \chi_\rho\nabla\phi\rangle_{L^2}+\llangle R_g, \langle\phi,\chi_\rho\phi\rangle\rrangle\\
 &\ge  -\langle \D\phi ,\chi_\rho \D\phi\rangle_{L^2}  +  \langle\nabla\phi, \chi_\rho\nabla\phi\rangle_{L^2}, 
 \end{align*}
where $\chi_\rho$ is the cut-off function \eqref{cut-off} 
defined in the proof of Proposition \ref{mass} (with $\eps=1$) and we used the non-negativity of $R_g$. Since the integrand on the left-hand side is $L^1_{-2q-1}$, Lemma \ref{subconverge} tells us that the left-hand side subconverges to zero as $\rho\to +\infty$. Meanwhile, we know that both terms on the right-hand side converge, and therefore 
$$ \| \nabla \phi \|_{L^2} \le  \| \D \phi \|_{L^2}. $$
Combining this with a weighted Poincar\'{e} inequality (see, for example, the more general statement \cite[Theorem 9.5]{Bartnik-Chrusciel}), we obtain the desired injectivity estimate
$$ \| \phi \|_{W^{1,2}_{-q}} \le  \| \D \phi \|_{L^2}. $$

We now only have to prove surjectivity. Given any $\eta\in L^2(S)$, we need to find a spinor $\xi\in W^{1,2}_{-q}(S)$ solving $\D\xi=\eta$.
We have shown that the pairing $\langle \omega,\phi\rangle_H:=\langle \D\omega ,\D\phi\rangle_{L^2}$ is equivalent to the $W^{1,2}_{-q}$ Hilbert product of $\omega$ and $\phi$. Applying the Riesz representation theorem to the functional $\phi\mapsto \langle \eta,\phi\rangle_{L^2}$ for the Hilbert product $H$, there must exist some $\omega\in W^{1,2}_{-q}(S)$ with the property that
$$ \langle \D\omega ,\D\phi\rangle_{L^2} = \langle \eta,\phi\rangle_{L^2},$$
for every $\phi\in W^{1,2}_{-q}(S)$. We claim that $\xi=\D\omega$ is the desired solution. We know that $\xi\in L^2(S)$. To prove better regularity, let $\xi_j$ be a sequence of $W^{1,2}_{-q}$ spinors converging to $\xi$ in $L^2(S)$. For any test function $\phi\in W^{1,2}_{-q}(S)$, we obtain 
$$
\lim_{j\to +\infty} \langle \D\xi_j , \phi\rangle_{L^2} = \lim_{j\to +\infty}\langle \xi_j , \D \phi\rangle_{L^2}
=\langle \xi , \D \phi\rangle_{L^2}=\langle \eta , \phi\rangle_{L^2},
$$
by construction of $\xi$. Therefore $\D\xi_j$ converges to $\eta$ in the weak $L^2$ topology. In particular, $\| \D\xi_j\|_{L^2}$ is bounded independently of $j$. The injectivity estimate then implies that $\|\xi_j\|_{W^{1,2}_{-q}}$ is bounded. Therefore $\xi_j$ must converge to $\xi$ weakly in $W^{1,2}_{-q}$, and we finish the argument by observing that 
\begin{align*}
\langle \D\xi , \phi\rangle_{L^2} = \langle \xi , \D\phi\rangle_{L^2}=\langle \eta , \phi\rangle_{L^2}
\end{align*}
for any compactly supported spinor $\phi\in W^{1,2}_{-q}(S)$.
\end{proof}

\begin{corollary}\label{Dirac-solve}
Under the assumptions in Theorem \ref{pmt}, let $\psi_0$ be a constant spinor defined over the asymptotically flat coordinate chart, as described earlier. Then there exists a spinor $\psi\in\Wloc^{1,2}(S)$ such that $\D\psi=0$ and $\psi-\psi_0\in W^{1,2}_{-q}(S)$.
\end{corollary}
\begin{proof}
It is easy to see that $\psi_0\in W^{1,2}_{-q}$, and thus $\D\psi_0 \in  W^{2}_{-q-1}(S)$. Apply the previous lemma to solve for $\xi$ in  $\D\xi=-\D\psi_0$, and then $\psi=\xi+\psi_0$ is the desired spinor.
\end{proof}

\begin{proof}[Proof of Theorem \ref{pmt}]
Now we simply put together all of our ingredients. Choose $\psi$ as in Corollary \ref{Dirac-solve} with $|\psi_0|=1$, and insert this $\psi$ into Lemma \ref{Lem:BoundaryTerm}, using Lemma \ref{BoundaryLimit} to simplify the annular integral on the left-hand side. Those three results combine to tell us that, for any $\eps>0$, 
\begin{align*}
\frac{1}{4\eps}\int_{\rho<r<\rho+\eps} V\cdot\overline\nabla r\,d\mu_h
& = -\langle \D\psi ,\chi_\rho \D\psi\rangle_{L^2}  +  \langle\nabla\psi, \chi_\rho\nabla\psi\rangle_{L^2}
\\
&\quad+  {1 \over 4} \llangle R_g, \langle\psi,\chi_\rho\psi\rangle\rrangle + \int_{\rho<r<\rho+\eps} u\,d\mu_h
\\
&= \langle\nabla\psi, \chi_\rho\nabla\psi\rangle_{L^2}
+  {1 \over 4} \llangle R_g, \langle\psi,\chi_\rho\psi\rangle\rrangle 
+ \int_{\rho<r<\rho+\eps} u\,d\mu_h, 
\end{align*}
where $\chi_\rho$ is the cut-off function \eqref{cut-off}. 
The first two terms on the right are non-negative and the last term subconverges to zero by Lemma \ref{subconverge}. Therefore, we have  
$$ 
\limsup_{\rho\to +\infty} \frac{1}{4\eps}\int_{\rho<r<\rho+\eps} V\cdot\overline\nabla r\,d\mu_h
= \| \nabla\psi\|^2_{L^2} +  {1 \over 4} \lim_{\rho\to +\infty} \llangle R_g, \langle\psi,\chi_\rho\psi\rangle\rrangle 
 \ge 0.
$$
Since $R_g$ is non-negative, we can apply Part 2 of Proposition \ref{mass} to see that the limit on the left exists (though is possibly infinite), and the result is proved.

For the rigidity, suppose that the mass is zero. The equality above then tells us that $\nabla\psi=0$. Since we can choose $\psi_0$ to be any constant spinor in the asymptotically flat coordinate chart, we can construct an entire basis of parallel spinors, which means that we have a parallel frame of the tangent bundle, which is only possible if $(M,g)$ is covered by Euclidean space, and the topology of $M$ ensures that the cover is trivial.  
\end{proof}

%=====================================================================

\section{Relation to earlier results} 
\label{sec:6}

In this final section, we show that our theory allows us to recover the earlier results for ``piecewise regular'' metrics by Miao \cite{Miao}, Shi and Tam \cite{ST}, and Lee \cite{Lee} in the spin case. 
Furthermore, it immediately follows from the Sobolev embedding theorem that, 
as far as comparatively more regular metrics are concerned,  
our Theorem \ref{pmt} clearly also generalizes \cite[Theorem 6.3]{Bartnik} (which assumes a $W^{2,p}_{-q}$ regular and asymptotically flat metric, with $p>n$ and $q\ge (n-2)/2$) and \cite{Grant-Tassotti} (which assumes a $\Wloc^{2,n/2}$ metric that is $C^2$ and asymptotically flat outside a compact set).

In view of the following proposition (by taking $a=0$ in the statement below), from our Theorem \ref{pmt} we immediately recover Theorem 1 in \cite{Miao}, Theorem 3.1 in \cite{ST}, and the main theorem in \cite{MFS} in the spin case. 
Our sign convention here is such that the mean curvature of a round sphere in Euclidean space points inward. 

\begin{proposition}[Metrics that are singular along a hypersurface]
\label{gluing}
Let $M_1$ and $M_2$ be smooth $n$-manifolds with boundaries, carrying $C^2$ Riemannian metrics $g_1$ and $g_2$, respectively.  Assume that there is an isometry $\Phi:(\partial M_1,g_1)\too(\partial M_2,g_2)$.  Let $(M,g)$ be the manifold obtained by
 gluing $(M_1,g_1)$ and $(M_2,g_2)$ along $\Phi$.  Let $\Sigma$ be the identification\footnote{Note that $g$ is Lipschitz continuous globally on $M$, and is $C^2$ \emph{up to} the hypersurface $\Sigma$ but not necessarily so \emph{across} $\Sigma$.}
 of $\partial M_1$ and $\partial M_2$ in $M$.  
Let $H_1$ and $H_2$ be the mean curvature vectors of $\Sigma$, computed with respect to the metrics $g_1$ and $g_2$, respectively.  Assume that $g_1$ and $g_2$ both have scalar curvature bounded below by $a$, and that at each point of $\Sigma$, $H_1-H_2$ either points into $M_1$ or is zero.  Then, in the sense of Definition~\ref{def:22}, $g$ has distributional scalar curvature bounded below by $a$, that is, 
$R_g \geq a$. 
\end{proposition}

\begin{proof}
Assume the hypotheses of Proposition \ref{gluing}.  Clearly, the distributional scalar curvature of $g$ is greater or equal to $a$ away from $\Sigma$.  We need only consider what happens at $\Sigma$. We choose smooth Fermi coordinates  $(x_0,x_1,\ldots,x_{n-1})$ on a small open ball intersecting $\Sigma$ and, in these coordinates, we have
\[
g(x_0,x_1,\ldots,x_{n-1})=dx_0^2+g_{r}(x_1,\ldots,x_{n-1}),
\]
where $g_{r}$ is a metric on the hypersurface $\Sigma_r$ whose signed distance $\Sigma$ is $r=x_0$.  We also choose our coordinates so that the region $x_0<0$ corresponds to $M_1$ and the region $x_0>0$ corresponds to $M_2$.

Let $h$ be the background metric given by the expression $dx_0^2+dx_1^2+\cdots+dx_{n-1}^2$.  Let $u$ be a smooth non-negative function supported in the coordinate ball described above.  Let $M_\epsilon$ be $M$ with an $\epsilon$-neighborhood of $\Sigma$ removed from it.  Let $\nu$ be dual form to the outward pointing normal of $\partial M_\epsilon$.
Then, by our Definition~\ref{def:22}, we find 
\begin{align*}
\llangle R_g,  u \rrangle
&= \int_M \left( -V\cdot\nablab \left( u \frac{d\mu_g}{d\mu_h}\right)  + F \left( u \frac{d\mu_g}{d\mu_h}\right) \right)\,d\mu_h\\
&= \lim_{\epsilon\to0} 
\int_{M_\epsilon} \left( -V\cdot\nablab \left( u \frac{d\mu_g}{d\mu_h}\right)  + F \left( u \frac{d\mu_g}{d\mu_h}\right) \right)\,d\mu_h, 
\end{align*}
thus 
\begin{align*}
\llangle R_g,  u \rrangle
& = \lim_{\epsilon\to0} \int_{M_\epsilon} R_g u\,d\mu_g
-\int_{\partial M_\epsilon} (V\cdot\nu) \left( u \frac{d\mu_g}{d\mu_h}\right)  \,d\sigma_h\\
& \ge \lim_{\epsilon\to0} \int_{M_\epsilon} a u\,d\mu_g + \int_{\partial M_{\epsilon}}
\hskip-.2cm \left(-\Gamma^k_{ij} g^{ij}\nu_k+\Gamma^k_{ki}g^{ij}\nu_j\right)\Big( u \frac{d\mu_g}{d\mu_h}\Big) dx_1\ldots dx_{n-1}.
\end{align*}
Now consider the integral over $\partial M_{\epsilon}$ broken down into two pieces over $\Sigma_\epsilon$ and $-\Sigma_\epsilon$.  Observe that
\begin{align*}
& \int_{\Sigma_{\epsilon}}\left(-\Gamma^k_{ij} g^{ij}\nu_k+\Gamma^k_{ki}g^{ij}\nu_j\right)\left( u \frac{d\mu_g}{d\mu_h}\right)\,dx_1\ldots dx_{n-1}
\\
&= \int_{\Sigma_{\epsilon}}\left(-\Gamma^0_{ij} g^{ij}+\Gamma^k_{ki}g^{i0}\right)\left( u \frac{d\mu_g}{d\mu_h}\right)\,dx_1\ldots dx_{n-1}
\\
&= \int_{\Sigma_{\epsilon}}\left( -\frac{1}{2}(\nablab_i g_{j0}+\nablab_j g_{i0} -\nablab_0 g_{ij})g^{ij} \right.
\\
& \qquad\quad \left.
\quad +\frac{1}{2}g^{kl}(\nablab_k g_{0l} + \nablab_0 g_{kl} -\nablab_l g_{k0})\right)\left( u \frac{d\mu_g}{d\mu_h}\right)\,dx_1\ldots dx_{n-1}
\\
&= \int_{\Sigma_{\epsilon}}g^{ij}\nablab_0 g^{ij}
\left( u \frac{d\mu_g}{d\mu_h}\right)\,dx_1\ldots dx_{n-1}
= \int_{\Sigma_{\epsilon}}-2\widetilde H_\epsilon
\left( u \frac{d\mu_g}{d\mu_h}\right)\,dx_1\ldots dx_{n-1}, 
\end{align*}
where $\widetilde H_\epsilon$ is the mean curvature (scalar) of $\Sigma_\epsilon$, taken with respect to the unit normal $\frac{\partial}{\partial x_0}$.  After performing a similar computation for $\Sigma_{-\epsilon}$ and then taking the limit as $\epsilon$ approaches zero, we conclude that
\begin{align*}
\llangle R_g,  u \rrangle&\ge \int_M au\,d\mu_g + \int_\Sigma 2 \, (\widetilde H_1-\widetilde H_2)\left( u \frac{d\mu_g}{d\mu_h}\right)  \,dx_1\ldots dx_{n-1}\\
&\geq \int_M au\,d\mu_g,
\end{align*}
since $(\widetilde H_1-\widetilde H_2) \geq 0$ by hypothesis.  This completes the proof of Proposition \ref{gluing}.
\end{proof}

Finally, in view of the following proposition, from our Theorem \ref{pmt} we also recover the main result in \cite{Lee} in the spin case. In fact, we verify a conjecture in that paper, as now stated.

\begin{proposition}[Metrics with ``small'' singular set]
Let $g$ be a Lipschitz continuous metric on a $n$-manifold $M$ such that $g$ is $C^2$ regular away from a closed set $S$ with vanishing $(n-1)$-dimensional Hausdorff measure. If $g$ has scalar curvature $R_g\ge a$ away from $S$, then $g$ has distributional scalar curvature $R_g\ge a$ on $M$  in the sense of Definition~\ref{def:22}. 
\end{proposition}

\begin{proof} Assume the hypotheses of the proposition. Observe that the $C^2$ regularity is defined on the {\sl open} set $M \setminus S$, and the scalar curvature is then well-defined in a classical sense on $M \setminus S$. 
Fix some $\epsilon>0$.
By the assumption on Hausdorff measure, the set $S$ can be covered by countably many balls $B_{\delta_i}(x_i)$ such that $\sum \delta_i^{n-1}<\epsilon$. Let $E_\epsilon = \bigcup_{1=1}^\infty B_{\delta_i}(x_i)$. The set $E_\epsilon$ is then a set with finite perimeter less than $\sum_{i=1}^\infty |\partial B_{\delta_i}(x_i)| = \sum_{i=1}^\infty \omega_{n-1}\delta_i^{n-1} < \omega_{n-1}\epsilon$, where $\omega_{n-1}$ is the volume of the standard unit $(n-1)$-sphere. In particular, $|\partial E_\epsilon|\to 0$ as $\epsilon\to0$, and clearly, $|E_\epsilon|\to0$.

For every smooth non-negative function $u$ with compact support, we have 
\begin{align*}
\llangle R_g,  u \rrangle
&= \int_M \left( -V\cdot\nablab \left( u \frac{d\mu_g}{d\mu_h}\right)  + F \left( u \frac{d\mu_g}{d\mu_h}\right) \right)\,d\mu_h\\
&= \int_{E_\epsilon}  \left( -V\cdot\nablab \left( u \frac{d\mu_g}{d\mu_h}\right)  + F \left( u \frac{d\mu_g}{d\mu_h}\right) \right)\,d\mu_h
\\
& \qquad 
+ \int_{M\smallsetminus E_\epsilon}  \left( -V\cdot\nablab \left( u \frac{d\mu_g}{d\mu_h}\right)  + F \left( u \frac{d\mu_g}{d\mu_h}\right) \right)\,d\mu_h. 
\end{align*}
Since $|E_\epsilon|\to0$ and the field $V$ (defined in \eqref{defineV}) is bounded for Lipschitz continuous metrics, the first integral converges to zero as $\epsilon\to0$. By the divergence theorem, the second integral term becomes ($d\sigma_h$ denoting the $(n-1)$-volume form induced by the background metric $h$) 
$$
\aligned
& \int_{M\smallsetminus E_\epsilon} R_g u\,d\mu_g + \int_{\partial E_\epsilon} (V\cdot\nu) \left( u \frac{d\mu_g}{d\mu_h}\right)  \,d\sigma_h
\\
& \ge \int_{M\smallsetminus E_\epsilon} a u\,d\mu_g + \int_{\partial E_\epsilon} (V\cdot\nu) \left( u \frac{d\mu_g}{d\mu_h}\right)  \,d\sigma_h,
\endaligned
$$
which approaches $\int_M a u\,d\mu_g$ as $\epsilon\to0$ since $|\partial E_\epsilon|\to 0$ and the field $V$ is bounded. 
\end{proof}

%======================================================================

\section{Application to the universal positive mass theorem}

The recent observations by Herzlich \cite{Herzlich} are now generalized, so that the non-negative curvature condition can be replaced by our non-negative distributional curvature condition. The modifications are essentially straightforward, so we will try to be brief. First of all, we must summarize the relevant parts of \cite{Herzlich}.

Choose any representation $V$ of $\spin(n)$ and then decompose $\rr^n\otimes V = \bigoplus_{j=1}^N W_j$ into its irreducible components. Given a spin structure on a manifold $M=M^n$ we may define the spinor bundles $E$ and $F_j$ associated with the representations $V$ and $W_j$, so that  
\[ 
T^*M \otimes E = \bigoplus_{j=1}^N F_j. 
\]
Given a metric on $M$, we can define a spin connection $\nabla$ on $E$. Let $\Pi_j$ be the projection onto the $F_j$ component, and define $P_j = \Pi_j \circ \nabla$. A {\bf generalized Bochner--Weitzenb\"{o}ck formula} is then defined 
as a choice of constants $a_j$ such that
\begin{equation} \label{BW}
\mathcal{R}:= \sum_{j=1^n} a_j P_j^*P_j 
\end{equation}
happens to be a curvature operator in $\End(E)$. This corresponds to choices of $a_j$  that make this expression a zero-th order operator. Given $V$, these solutions are classified in the literature \cite{Homma,SW}. 
Next, we introduce 
$$
P_+ = \sum_{a_j>0} \sqrt{a_j}P_j, \qquad P_+ = \sum_{a_j>0} \sqrt{a_j}P_j.
$$ 

The main result of \cite{Herzlich} can be stated as follows. 

\begin{theorem}[Herzlich's universal positive mass theorem]
\label{thm:Herzlich}
Let $(M,g)$ be a smooth, complete, asymptotically flat spin manifold. Then, with the construction above, there exists a constant $\mu$ depending only on $V$ and $a_j$ such that if $\psi_k$ are sections of $E$ asymptotic to an orthonormal basis of $V$ at infinity (with appropriate decay), then one has the following expression of the ADM mass of $M$:  
\begin{equation}\label{massformula}
 \mu(V,a_j) \, \mADM = \sum_{k=1}^{\dim V}  -\| P_+\psi_k\|_{L^2}  + \|P_- \psi_k\|^{L^2} + \langle \mathcal{R}\psi_k ,\psi_k \rangle_{L^2}.
 \end{equation}
\end{theorem}

Clearly, in view of the formula \eqref{BW} on $\langle \mathcal{R}\psi_k ,\psi_k \rangle_{L^2}$ and by suitably integrating by parts, one must arrive at a formula like \eqref{massformula}. The important observation is that the boundary term is precisely a \emph{multiple $\mu$ of the mass,} at least after summing over an orthonormal basis and  $\mu$ is a {\sl universal constant} independent of $(M,g)$ (and computed explicitly in \cite{Herzlich}).
 From Theorem \ref{thm:Herzlich} it is clear that if $\mathcal{R}$ is a non-negative operator and one can find appropriate solutions of $P_+ \psi_k=0$, then one obtains either $\mADM\ge0$ or $\mADM\le0$ (depending on the sign of $\mu$). 

We now claim that there is a distributional version of Theorem \ref{thm:Herzlich}. Consider the basic setup that was used in Sections 2 through \ref{sec:6}. 
Given a curvature operator $\mathcal{R}$ as in \eqref{BW}, it must be of the form 
\be
\mathcal{R} = \overline{\Div} V + F,
\ee
where $V$ is a section of $TM\otimes \End(E)$ and $F$ is a section of $\End{E}$, and these sections must have the same respective regularity as the terms $V$ and $F$ that appear in the scalar curvature decomposition formula \eqref{scalar}. This allows us to define $\mathcal{R}$ as a distribution as in Definition \ref{def:22}, with obvious modification due to the fact that $\mathcal{R}$ is no longer a scalar. Just as in Proposition \ref{distLW}, we easily obtain a distributional version of the generalized Bochner--Weitzenb\"{o}ck formula \eqref{BW}.

\begin{lemma}[Universal Lichnerowicz--Weitzenb\"{o}ck formula for distributional curvature] 
Assume that $g$ is a $C^0\cap \Wloc^{1,n}$ metric on a smooth spin manifold $M^n$. If $\psi$ and $\phi$ are $\Wloc^{1,2}$ spinors and $\phi$ has compact support, then one has 
\bel{LW}
0= - \langle P_+\psi , P_+ \phi\rangle_{L^2} +  \langle P_- \psi, P_-\phi\rangle_{L^2}
+\left\llangle \mathcal{R}_g \psi,\phi\right\rrangle,
 \ee
 where the last term is defined in the sense of distributions.
\end{lemma}

Our next step is to compute the boundary term as in Lemma 2.2 of \cite{Herzlich}, except that we must modify the computation along the lines of Section \ref{sec:32}.

\begin{lemma}[A universal identity for distributional curvature] 
Let $(M,\Phi)$ be a background $n$-manifold endowed with background metric data $(h,r)$, and assume that $g$ is a $C^0\cap W^{1,n}_{-q}$ regular, asymptotically flat metric with $q=(n-2)/2$. Let $\psi_0$ be a constant\footnote{with respect to the flat connection, as explained in \cite{Herzlich}} spinor in $E$ defined over the asymptotically flat coordinate chart. Then for  any $\rho>2$, $\eps>0$, and any spinor $\psi$ with $\psi-\psi_0 \in W^{1,2}_{-q}(E)$, one has 
\be\label{GeneralBoundary}
\aligned
& - \langle P_+\psi ,\chi_\rho P_+\psi \rangle_{L^2} +  \langle P_- \psi, \chi_\rho P_-\psi\rangle_{L^2}
+\left\llangle \mathcal{R}_g \psi,\chi_\rho\psi\right\rrangle 
\\
&= -\frac{1}{\eps}\int_{\rho<r<\rho+\eps} \left\langle (\overline{\nabla} r)\otimes \psi_0, \sum_{j=1}^N a_j\Pi_j(\nabla \psi_0)\right\rangle\,d\mu_g + \int_{\rho<r<\rho+\eps} u\,d\mu_h 
\endaligned
\ee
for some $u\in L^1_{-2q-1}$, where $\chi_\rho$ is the cut-off function introduced in Section \ref{sec:32}.
\end{lemma}

Finally, the same argument as in \cite[Theorem 3.1]{Herzlich} shows that if we sum the above formula over a basis, then the annular boundary term in \eqref{GeneralBoundary} must be a multiple of the ``classical'' case in which we have already shown that the limit of the annular term corresponds to the mass, and this multiple must depend only on $V$ and $a_j$. 

Finally, if we assume that $R_g$ is a finite, signed measure outside a compact set (so that we can invoke Part 1 of Proposition \ref{mass}) and make a similar assumption about $\mathcal{R}$, we obtain a distributional version of equation \eqref{massformula} with our weaker regularity assumptions.

%======================================================================

\section*{Acknowledgments}

This material is based upon work supported by the National Science Foundation under Grant No.\ 0932078 000, while the authors were in residence at the Mathematical Sciences Research Institute in Berkeley, California, during the Fall of 2013. The second author (PLF) was partially supported by the Agence Nationale de la Recherche through the grant ANR SIMI-1-003-01. 

%======================================================================

\end{document}